\documentclass[a4paper, amssymb, amsmath, reprint, showkeys, nofootinbib, twoside]{revtex4-1}

\usepackage{amsthm}
\usepackage{mathtools}
\usepackage{physics}
\usepackage{mathrsfs}
\usepackage{xcolor}
\usepackage{graphicx}
\usepackage[left=25mm,right=20mm,top=20mm,columnsep=17pt]{geometry} \usepackage{adjustbox}
\usepackage{placeins}
\usepackage[T1]{fontenc}
\usepackage{lipsum}
\usepackage[utf8]{inputenc}
\usepackage{csquotes}
\usepackage[colorlinks=true,citecolor=blue,linkcolor=blue]{hyperref}%
\usepackage[english]{babel}
\usepackage{float}
\usepackage{tikz}
\usetikzlibrary{shapes, shapes.geometric, shapes.symbols, shapes.arrows, shapes.multipart, shapes.callouts, shapes.misc,decorations.pathmorphing}

\tikzset{snake it/.style={decorate, decoration=snake}}

\usepackage{tikz-3dplot}

\theoremstyle{definition}
\newtheorem{dfn}{Definition}

\theoremstyle{plain}

\newtheorem{lemma}{Lemma}

\newcommand{\ri}{\rightarrow}
\newcommand{\muu}{\mu}
\DeclareMathOperator{\ran}{Ran}

\begin{document}

\title{Contextuality-by-default for behaviours in compatibility scenarios}

\author{Alisson Tezzin}
    \email[Correspondence email address: ]{alisson.tezzin@usp.br}
    \affiliation{University of São Paulo, Institute of Physics, São Paulo, SP, Brazil}
\author{Rafael Wagner}
\affiliation{University of São Paulo, Institute of Physics, São Paulo, SP, Brazil}
\author{Bárbara Amaral}
\affiliation{University of São Paulo, Institute of Physics, São Paulo, SP, Brazil}

\date{\today} 

\begin{abstract}

We show that the main idea behind contextuality-by-default (CbD), i.e., the assumption that a physical measurement has to be understood as a contextual collection of random variables, is implicit in the compatibility-hypergraph approach to contextuality (CA) and use this result to develop in the latter important concepts which were introduced in the former. We introduce in CA the non-degeneracy condition, which is the analogous of consistent connectedness, and prove that this condition is, in general, weaker than non-disturbance condition. The set of non-degenerate behaviours defines a polytope, implying that one can characterize consistent connectedness  using linear inequalities. We introduce the idea of extended contextuality for behaviours and prove that a behaviour is non-contextual in the standard sense iff it is non-degenerate and non-contextual in the extended sense. Finally, we use extended scenarios and behaviours to shed new light on our results.
\end{abstract}


\maketitle

\section{Introduction}

Contextuality is one of the most singular characteristics of Quantum Theory. In addition to its role in the search for a deeper understanding of the theory itself \cite{ACC14, amaral2014exclusivity}, recent work provides indisputable evidence that contextuality is an essential resource in various information protocols and computational tasks \cite{NoncontextualWirings}.
Contextuality is a property displayed by the statistics of measurements performed on a quantum system which shows that such a statistic is incompatible with the expected description for classical systems \cite{amaral2018graph}. This characteristic is closely related to the existence of incompatible measurements in quantum systems.

One approach which allows us to deal with this concept is the \textbf{compatibility-hypergraph approach (CA)} \cite{amaral2018graph,CSW2014,amaral2014exclusivity}, whose main elements are  (compatibility) scenarios and behaviours. The definition of scenario is the following. 

\begin{dfn}[Scenario]\label{def: DefinitionScenario} A scenario is a triple $\mathcal{S} \equiv (\mathcal{X},\mathcal{C},O)$ where $\mathcal{X}$, $O$ are finite sets and $\mathcal{C}$ is a collection of subsets of $\mathcal{X}$ satisfying the following properties.
\begin{itemize}
\item[(a)] $\mathcal{X} = \cup \mathcal{C}$ \label{property a}
\item[(b)] For $C,C' \in \mathcal{C}$, $C' \subset C$ implies $C' = C$
\end{itemize}
\end{dfn}
We call the elements of $\mathcal{X}$ \textbf{measurements} and the elements of $\mathcal{C}$ \textbf{contexts}. The picture in mind is that $\mathcal{X}$ represents a collection of measurements which we can perform over a given experiment and that a context represents a collection of compatible measurements; the set $O$ represents the set of all possible outcomes for the measurements in the scenario. The approach is called "hypergraph-approach" because a scenario $(\mathcal{X},\mathcal{C},O)$ can be associated to a hypergraph whose vertices are the elements of $\mathcal{X}$ and whose hyperedges are the elements of $\mathcal{C}$.


The result of a joint measurement over a context $C$ can be represented by a function $C \ri O$. Therefore, the set $O^{C}$ of all functions $C \ri O$ can be understood as the set of all possible outcomes of a joint measurement on $C$. If it is possible to perform repetitive joint measurements over a context we should be able to assign probabilistic predictions to the resulting outcomes, henceforth providing an interpretation for $O^C$; this construction is known as a behaviour of the scenario \cite{amaral2018graph,amaral2019extendedcontextuality,amaral2018necessaryextended}.

\begin{dfn}[Behaviour]\label{def: behaviour}
Let $\mathcal{S}$ be a scenario. A behaviour in $\mathcal{S}$ is a function $p$ which associates to each context $C$ a probability distribution $p^{C}$ on $O^{C}$, that is, $p^{C}:O^{C} \rightarrow [0,1]$ satisfies $\sum_{s \in O^{C}} p^{C}(s) = 1$. 
\end{dfn}

In CA it is usually assumed that behaviours which represents probability assignments with physical meaning should satisfy the so called non-disturbance condition, which imposes that the probability distributions given by a behaviour must coincide in intersections of contexts: 

\begin{dfn}[Non-disturbance]\label{DefNondisturbance}
A behaviour $p$ in a scenario $\mathcal{S}$ is said to be non-disturbing if the condition $p^{C}_{C \cap D} = p^{D}_{C \cap D}$ holds for any intersecting contexts $C,D \in \mathcal{C}$, where
\[p^{E}_{E'}(t) \doteq \sum_{\begin{subarray}{l} s \in O^{E} \\ s|_{E'} = t \end{subarray}} p^{E}(s)\ \ \ \ \ \ \ \ \forall t \in O^{E'}\]
for any $E \in \mathcal{C}$ and $E' \subset E$. 
\end{dfn}

In this paper we will see that, in order to introduce contextuality-by-default in CA, we have to give up this requirement.

Another important physical situation happens when all measurements in a scenario are compatible. Such situations are consistent with global probability assignments to measurements. The standard definition of contextuality constitutes an attempt to deal with it:

\begin{dfn}[Standard contextuality]\label{DefNoncontextualStandard}
    A behaviour $p$ in a scenario $\mathcal{S}$ is said to be non-contextual if there is a probability distribution $\overline{p}:O^{\mathcal{X}} \ri [0,1]$ satisfying, for any context $C$,
    \[p^{C} = \overline{p}_{C},\]
    where $\overline{p}_{C}$ denotes the restriction of $\overline{p}$ in $O^{C}$, that is, for any $s \in O^{C}$,
    \[\overline{p}_{C}(s) \doteq \sum_{\begin{subarray}{l} t \in O^{\mathcal{X}} \\ t|_{C} = s\end{subarray} }p^{C}(t).\]
\end{dfn}

In any scenario, non-contextual behaviours are always non-disturbing. On the other hand, it is in general false that non-disturbing behaviours are non-contextual \cite{amaral2018graph}.

Another approach to deal with contextuality is what we will call the \textbf{contextuality-by-default approach (CbD)} \cite{dzhafarov2016contextualitybydefault, kujala2015necessary}. In this approach we consider random variables instead of behaviours, and the main idea is that a physical measurement has to be seen as a collection of random variables, one for each context containing the measurement. The approach is structured as follows. We consider a finite set $\mathscr{X}$, whose elements are called \textit{properties}; this set represents a collection of physical properties of the physical system under description. From an operational point of view we can understand ``measurement'' and ``physical property'' as analogous concepts, i.e. the set $\mathscr{X}$ fulfills the same role as the set $\mathcal{X}$ of measurements in a scenario. In this approach it is also assumed that there are (physical) properties which are incompatible - that is the point of contextuality - and therefore we take a collection $\mathscr{C}$ of subsets of $\mathscr{X}$, whose elements are also called contexts. To any context $C \in \mathscr{C}$ we associate a collection $\mathscr{R}^{C} \equiv\{R^{C}_{x}; x \in C\}$ of random variables; this collection is said to be ``the result of jointly measuring all properties within C'' \cite{kujala2015necessary}. By doing this we associate to each property $x \in \mathscr{X}$ a collection $\mathscr{R}_{x} \equiv\{R^{C}_{x}; C \in \mathscr{C}_{x}\}$ of random variables, where $\mathscr{C}_{x}$ denotes the set of contexts containing $x$. The sets $\mathscr{R}_{x}$, $x \in \mathscr{X}$, are called \textbf{connections} \cite{kujala2015necessary}. We impose that all the random variables $R^{C}_{x}$, for any $x \in \mathscr{X}$ and any $C$ containing $x$, have the same codomain. Finally, we call  \textbf{system} the triple $(\mathscr{X},\mathscr{C},\mathscr{R})$, where $\mathscr{R} \doteq \cup_{C \in \mathscr{C}} \mathscr{R}^{C}$. 

A system $(\mathscr{X},\mathscr{C},\mathscr{R})$ is said to be \textbf{consistently connected} when, for any  $x \in \mathscr{X}$, all the random variables in $\mathscr{R}_{x}$ have the same distribution (that is, are physically equivalent in the sense of appendix \ref{appendix: couplings}); if this property does not hold we say that the behaviour is inconsistently connected. Finally, a system $(\mathscr{X},\mathscr{C},\mathscr{R})$ is said to have a \textbf{maximally non-contextual description} if there is a coupling $S$ of $\mathscr{R}$ satisfying the following property: for any $x \in \mathscr{C}$, the restriction of $S$ to $\mathscr{C}_{x} \equiv \{C \in \mathscr{C}; x \in C\}$ is a maximal coupling of the connection $\mathscr{R}_{x}$ (see appendix \ref{appendix: couplings} for the definition of coupling). We refer to this notion of contextuality as \textbf{contextuality in the extended sense}.

\begin{figure}[H]
\centering
\scalebox{0.68}{ 
\begin{tikzpicture}

\draw[fill=yellow] (-6,-2) rectangle ++(4,4);
\draw[fill=blue] (-6,2)--(-4,4)--(-2,2)--(-6,2);
\draw[fill=red] (-6,-2)--(-4,-4)--(-2,-2)--(-6,-2);
\draw[fill=green!60] (-2,-2)--(0,0)--(-2,2)--(-2,-2);
\node at (1.5,-1.5) {$ND$};
\node at (2.32,-1.5) {$=$};
\draw[fill=yellow] (2.7,-1.5) circle (0.1cm);
\draw[fill=blue] (3.3,-1.5) circle (0.1cm);
\node at (1.5,-2) {$NC$};
\node at (2.32,-2) {$=$};
\draw[fill=yellow] (2.7,-2) circle (0.1cm);
\node at (1.75,-1) {$NC_{ext}=$};
\draw[fill=yellow] (2.7,-1) circle (0.1cm);
\draw[fill=green!60] (3.3,-1) circle (0.1cm);
\node at (1.75,-0.5) {$ND_{eg}=$};
\draw[fill=yellow] (2.7,-0.5) circle (0.1cm);
\draw[fill=blue] (3.3,-0.5) circle (0.1cm);
\draw[fill=red] (3.9,-0.5) circle (0.1cm);
\node at (3,-1.5) {$+$};
\node at (3,-1) {$+$};
\node at (3,-0.5) {$+$};
\node at (3.6,-0.5) {$+$};
\end{tikzpicture}
}
\caption{\label{politopo ND, NDeg, NCext, NC}Description of the different set inclusions between behaviours. In the picture we have $ND_{eg}$ as the polytope of non-degenerate behaviours, $NC_{ext}$ as the set of non-contextual behaviours with respect to the extended definition (definition \ref{DefinitionNoncontextual}); $NC$ represents the polytope of non-contextual behaviours in the standard sense (definition \ref{DefNoncontextualStandard}) and $ND$ illustrates the polytope of non-disturbing behaviours.}  
\end{figure}
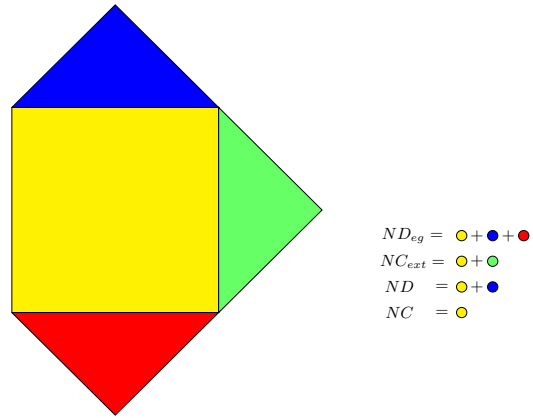

\subsection*{Outline of this paper}

In section \ref{SecBehavioursSystems} we establish a precise relation between behaviours and systems in a scenario: we show how a behaviour defines a system (which is unique up to ``physical equivalences'') and how we can associate a behaviour (not unique) to any system in a scenario. This relation make it clear that the compatibility-hypergraph approach  (CA) is entirely compatible with the idea behind contextuality-by-default, allowing us to develop in CA important concepts introduced in the contextality-by-default approach (CbD). In section \ref{SecNondegenerate} we introduce consistent connectedness in CA - we call non-degenerate a behaviour whose system associated to it is consistently connected. We prove that non-disturbance implies this condition and which, in general, the reverse is not true, i.e., we prove that non-disturbance and consistent connectedness are not equivalent concepts. The set of non-degenerate behaviours define a polytope, allowing us to characterize consistent connectedness using linear inequalities \cite{brondsted2012introduction}. In section \ref{SecNoncontextuality} we introduce the definition of extended contextuality for behaviours and prove that a behaviour is non-contextual in the standard sense (definition \ref{DefNoncontextualStandard}) iff it is non-contextual in the extended sense and non-degenerate. Figure \ref{politopo ND, NDeg, NCext, NC} depicts how these sets of behaviours relate to each other. Finally, in section \ref{SecExtended} we rewrite our results using extended scenarios and behaviours introduced in \cite{amaral2018necessaryextended}.

\section{Behaviours and systems}\label{SecBehavioursSystems}

In this paper we will study systems in a scenario, and we justify this decision as follows. First of all, it is easy to associate a finite set $O$ to a system: as we do in a scenario, this is the set of all possible outcomes of the properties  represented in the system; moreover, in \cite{kujala2015necessary} the authors make it clear, by means of examples and with a comparison with the ``traditional approach to contextuality'', that the random variable $R^{C}_{x}$  is a random variable with values on the set of outcomes $(O, \mathcal{P}(O))$ (we are using the terminology presented in the appendix \ref{appendix: couplings}). Furthermore it seems that in a system  $(\mathscr{X},\mathscr{C},\mathscr{R})$ it is implicitly assumed that the set of contexts $\mathscr{C}$ satisfies the property (a) of definition \ref{def: DefinitionScenario}. In fact, if $x \in \mathscr{X}$ does not belong to any context, there is no random variable associated to it, hence we cannot talk about a measurement of this property. Finally, can we assume that $\mathscr{C}$ is an anti-chain? Even though it seems to be the case, it is not completely clear to us. If that is not the case, the idea of ``system in a scenario'' can be understood at least as a particular case of system. 
\begin{dfn}[System in a scenario]\label{def: system in a scenario}
Let $\mathcal{S} \equiv (\mathcal{C}, \mathcal{C},O)$ be a scenario. A system in $\mathcal{S}$ is any system $(\mathcal{X},\mathcal{C},\mathscr{R})$ where, for any $x \in \mathcal{X}$ and $C \in \mathcal{C}$ containing $x$, $\ran(R^{C}_{x}) \subset O$.
\end{dfn}

Finally, let $(\mathcal{X},\mathcal{C},\mathscr{R})$, $(\mathcal{X},\mathcal{C},\mathscr{S})$ be systems in a scenario $(\mathcal{X},\mathcal{C},O)$. We say that  $(\mathcal{X},\mathcal{C},\mathscr{R})$ and $(\mathcal{X},\mathcal{C},\mathscr{S})$ are \textbf{physically equivalent} when $R^{C}_{x}$ and $S^{C}_{x}$ have the same distribution for any $C \in \mathcal{C}$ and $x \in C$. 

As it has been stated in the introduction, the key idea behind  CbD is that we associate to a property (measurement) not to one but to a collection of random variables, one for each context containing that property. We will see now that the definition of behaviour in CA is entirely compatible with this assumption: if we want to treat measurements as random variables whose joint distributions in contexts are given by some behaviour, we need to assume that a measurement is associated to a collection of random variables, one for each context containing the measurement. Before we begin our discussion let's fix some notation.

\begin{dfn}\label{DefinitionForNotation} Let $p$ be a behaviour in a scenario $\mathcal{S} \equiv (\mathcal{X},\mathcal{C},O)$. For $x \in \mathcal{X}$ we write $\mathcal{C}_{x} \doteq \{C \in \mathcal{C}; x \in C\}$. We define, for $C \in \mathcal{C}_{x}$, the distribution $p^{C}_{x}:O \ri [0,1]$ by 
$$p^{C}_{x}(o) \doteq \sum_{\begin{subarray}{l} s \in O^{C} \\ s(x) = o\end{subarray}} p^{C}(s) \ \forall o \in O.$$ 

We denote by $P$ the set of all this probability distributions and we define subsets $P^{C}$, $P_{x}$ as follows.
\begin{itemize}

    \item[(a)] For any $C \in \mathcal{C}$,
    $$P^{C} \doteq \left\{p^{C}_{x}; x \in C\right\}$$
    
    \item[(b)] For any $x \in \mathcal{X}$,
    $$P_{x} \doteq \left\{p^{C}_{x}; C \in \mathcal{C}_{x}\right\}.$$
    \item[(c)] Note that $P = \cup_{C \in \mathcal{C}} P^{C} = \cup_{x \in \mathcal{X}} P_{x}$.
\end{itemize}

    In case we want to refer to families of probability distributions instead of sets we will add an underline in the notation. For example, given a behaviour $p$ we define the following families of probability distributions in $(O,\mathcal{P}(O))$.
    \begin{align}
        &\underline{P} \doteq \left (p^{C}_{x}\,\, |\,\, \forall x \in \mathcal{X} (C \in \mathcal{C}_{x})\right)\\
        &\underline{P}^{C} \doteq \left(p^{C}_{x}; x \in C\right)
        \\
        &\underline{P}_{x} \doteq \left(p^{C}_{x}; C \in \mathcal{C}_{x}\right)
    \end{align}
\end{dfn}

    With definition \ref{DefinitionForNotation} we see that a behaviour naturally associates a collection $\underline{P}_{x}$ of probability distributions in $(O,\mathcal{P}(O))$ to each measurement $x \in \mathcal{X}$. It is in general false that $p^{C}_{x} = p^{C'}_{x}$ whenever $C,C' \in \mathcal{C}_{x}$ - for example, take a disturbing behaviour in a $n$-cycle scenario \cite{amaral2019extendedcontextuality}. Moreover, in each context $C$, $p^{C}$ is a coupling (see appendix \ref{appendix: couplings}) of $\underline{P}^{C}$. We know that any probability distribution in a finite set $O$ is the distribution (or density function) of a random variable with values on $(O,\mathcal{P}(O))$, therefore, denoting by $R^{C}_{x}$ a random variable whose distribution is $p^{C}_{x}$ (all the candidates are physically equivalent in the sense of see appendix \ref{appendix: couplings}) we see that a behaviour associates to a measurement $x \in \mathcal{X}$ a collection $\mathcal{R}_{x}$ of random variables $R^{C}_{x}$, one for each context $C$ containing $x$. That is exactly what we consider in the contextuality-by-default approach. As stated before, we conclude that the definition of behaviour is only coherent with the assumption that physical measurements are random variables if we accept that a measurement can be associated to different random variables, depending on the context in which we are measuring it. That is, the key idea behind CbD  is already implicit in the contextuality approach.

We have seen that a behaviour in $\mathcal{S}$ always defines one system  in $\mathcal{S}$ (actually, it defines a collection of physically equivalent systems). The \textbf{system associated to the behaviour $p$} is any system $(\mathcal{X},\mathcal{C},\mathscr{R})$ in $(\mathcal{X},\mathcal{C},O)$  such that $p^{C}_{x}$ is the distribution of $R^{C}_{x}$ for any $C \in \mathcal{C}$ and $x \in \mathcal{X}$. On the other hand, we can associate more then one behaviour in $\mathcal{S}$ to a given system in $\mathcal{S}$: any collection of contextual couplings of a system is a behaviour. Systems and behaviours are not equivalent definitions, as we should expect: this is due the the non uniqueness of couplings.

\section{Non-degenerate behaviours}\label{SecNondegenerate}

We call non-degenerate a behaviour which associates only one probability distribution to each measurement of the scenario: 
\begin{dfn}[non-degenerate behaviour] A behaviour $p$ in a scenario $\mathcal{S} \equiv (\mathcal{X},\mathcal{C},O)$ is said to be non-degenerate if for any $x \in \mathcal{X}$ the equality
\begin{equation}\label{non-degenerate equation}
p^{C}_{x} = p^{C'}_{x}
\end{equation}
holds whenever $C,C' \in \mathcal{C}_{x}$. Equivalently, a behaviour is non-degenerate when the system associated to it is consistently connected.
\end{dfn}
If $p$ is a non-degenerate behaviour, we denote by $p_{x}$ the distribution which $p$ defines for $x \in \mathcal{X}$, that is, $p_{x} \doteq p^{C}_{x}$, where $C$ is any context containing $x$. In this case we have $P = \{p_{x}; x \in \mathcal{X}\}$, $P^{C} = \{p_{x}; x \in C\}$ and $P_{x} = \{p_{x}\}$ (see definition \ref{DefinitionForNotation}).

By definition, the system associated to a non-degenerate behaviour is consistently connected. On the other hand, if $(\mathcal{X},\mathcal{C},\mathscr{R})$ is a consistently connected system in the scenario $(\mathcal{X},\mathcal{C},O)$, any behaviour associated to it, that is, any behaviour which is a contextual collection of couplings for $\mathscr{R}$, is non-degenerate.  

The set of non-degenerate behaviours draws the line between the standard interpretation of physical measurements (in which a measurement is seen as one random variable) and the interpretation proposed in the contextuality-by-default approach (a measurement is a contextual collection of random variables). This condition also has a physical content similar to non-disturbance condition. This suggests that the characterization of this set is physically relevant, just as the characterization of the non-disturbing set \cite{amaral2018graph}.

Such characterization is usually done in $\mathbb{R}^{N}$, so let's briefly explain how behaviours in a scenario $\mathcal{S} \equiv (\mathcal{X},\mathcal{C},O)$ are associated to elements of $\mathbb{R}^{N}$. First of all we define
\[N \doteq \sum_{C \in \mathcal{C}} \vert O ^{X} \vert,\]
where, for any finite set $A$, $\vert A\vert$ denotes the number of elements of $A$. Notice that $\vert O \vert^{\vert C\vert} = \vert O^{C} \vert$. Now we fix any bijective mapping $\cup_{C \in \mathcal{C}} O^{C} \rightarrow \{1, ..., N\}$ and denote by $(s|C)$ the image of $s \in O^{C}$ under it. If $x \in \mathbb{R}^{N}$, $x_{(s|C)}$ denotes the component $(s|C)$ of $x$. Finally, we associate a behaviour $p$ in $\mathcal{S}$ to the element $\phi(p) \equiv P \in \mathbb{R}^{N}$ defined by
\[P_{(s|C)} \doteq p^{C}(s)\]
for any $(s|C) \in \{1, ..., N\}$. We denote by $NC$ the set of all $\phi(p) \in \mathbb{R}^{N}$ such that $p$ is non-contextual in the standard sense; this is the so-called non-contextual (in the standard sense) set. Analogously we define the non-disturbing set $ND$, the non-degenerate set $ND_{eg}$ and the non-contextual in the extended sense set $NC_{ext}$. It is well known that $NC$ and $ND$ are polytopes in $\mathbb{R}^{N}$ \cite{araujo2013allnoncontextualityineq,amaral2018graph,amaral2014exclusivity}. We also note that $ND_{eg}$ is a polytope since it constitutes of a set of linear restrictions, i.e., hyperplanes defined by equation  \eqref{non-degenerate equation}, over the polytope of behaviours \cite{amaral2018graph,brondsted2012introduction}. 

Determining the inequalities which characterizes polytopes like $NC$ or $ND$ in a given scenario is a very important problem in the contextuality approach \cite{amaral2018graph}. Such inequalities provide very useful information about the nature of physical systems, with respect to the property defining the polytope. The polytope $ND_{eg}$ allows us to  do the same with non-degeneracy (i.e., consistent connectedness). 

 The non-disturbance condition imposes that the probabilities given by a behaviour coincides in intersections of contexts. The non-degeneracy condition, on the other hand, imposes that such distributions coincide in every point $x \in \mathcal{X}$. Non-disturbing behaviours are always non-degenerate (lemma \ref{LemmaNondisturbingNondegenerate}), but the reverse is false (lemma \ref{NDeg is not ND}). That is what we are going to prove now. 

\begin{lemma}\label{LemmaNondisturbingNondegenerate}
Any non-disturbing behaviour is non-degenerate.
\end{lemma}
\begin{proof}
Let $p$ be a non-disturbing behaviour in a scenario $\mathcal{S}$. We will prove that, for any nonempty subset $E$ of intersecting contexts $C,D \in \mathcal{C}$, we have $p^{C}_{E} = p^{D}_{E}$. In particular, it follows that $p^{C}_{x} = p^{D}_{x}$ for any $x \in \mathcal{X}$ and $C,D \in \mathcal{C}_{x}$ (just take $E = \{x\}$). So let $C,D$ be intersecting contexts (note that $C=D$ is a particular case) and take any nonempty set $E \subset C \cap D$; for  $r \in O^{E}$ we have
\begin{align}
p^{C}_{E}(r) &\doteq \sum_{\begin{subarray}{l} t \in O^{C} \\ t|_{E} = r\end{subarray}} p^{C}(t) = \sum_{\begin{subarray}{l} s \in O^{C\cap D} \\ s|_{E} = r\end{subarray}}\sum_{\begin{subarray}{l} t \in O^{C} \\ t|_{C \cap D} = s\end{subarray}} p^{C}(t)
\\
&= \sum_{\begin{subarray}{l} s \in O^{C\cap D} \\ s|_{E} = r\end{subarray}}p^{C}_{C\cap D}(s) = \sum_{\begin{subarray}{l} s \in O^{C\cap D} \\ s|_{E} = r\end{subarray}}p^{D}_{C\cap D}(s) 
\\
&= \sum_{\begin{subarray}{l} s \in O^{C\cap D} \\ s|_{E} = r\end{subarray}}\sum_{\begin{subarray}{l} u \in O^{D} \\ u|_{C \cap D} = s\end{subarray}} p^{D}(u)
\\
&= \sum_{\begin{subarray}{l} u \in O^{D} \\ u|_{E} = r\end{subarray}}p^{D}(u) = p^{D}_{E}(r),
\end{align}
and this implies that $p^{C}_{E} = p^{D}_{E}$
\end{proof}

Now let's see why non-degenerate behaviours are not necessarily non-disturbing. Let $p$ be a non-degenerate behaviour in $\mathcal{S}$ and $P \doteq \{p_{x}; x \in \mathcal{X}\}$ the collection of probability distributions in $O$ defined by it. By definition, if $C \cap D \neq \emptyset$, the marginals $p^{C}_{C \cap  D}$ and $p^{D}_{C \cap D}$ are couplings of $p_{x}$, $x \in C \cap D$. Consequently, saying that $p$ is non-disturbing means that $p^{C}_{C \cap D}$ and $p^{D}_{C \cap D}$ are necessarily the same coupling. But it is easy do find a non-degenerate behaviour which does not satisfies this property. In fact, take the scenario $\mathcal{S} \equiv (\mathcal{X},\mathcal{C},O)$ where $\mathcal{X} = \{a,b,c,d\}$ has four elements, $\mathcal{C} \doteq \{C,D\}$ for $C \doteq \{a,b,c\}$, $D\doteq \{b,c,d\}$, and $O$ is any finite set. For each $x \in \mathcal{X}$ we associate a  probability distribution $p_{x}:O \ri [0,1]$ such that the sequence $p_{b},p_{c}$ has more then one coupling and we take two different couplings $f,g: O^{2}  \ri [0,1]$ of it.  We define $p^{C}:O^{C} \ri [0,1]$ and $p^{D}:O^{D} \ri [0,1]$ by
\begin{align}
    p^{C}(s) &\doteq p_{a}(s(a))f(s(b),s(c))
    \\
    p^{D}(t) &\doteq f(s(b),s(c))p_{d}(s(d)).
\end{align}
These are probability distributions, therefore they define a behaviour $p$ in $\mathcal{S}$. By construction, $p$ is disturbing and non-degenerate. In fact,
\begin{align}
    p^{C}_{C \cap D}(r) &= \sum_{\begin{subarray}{l}s \in O^{C} \\ s|_{C \cap D} = r\end{subarray}}p^{C}(s) 
    \\
    &= \sum_{\begin{subarray}{l}s \in O^{C} \\ s|_{C \cap D} = r\end{subarray}}p_{a}(s(a))f(r(b),r(c))
    \\
    &= f(r(b),r(c)),
\end{align}
and analogously $p^{D}_{C \cap D}(r) = g(r(b),r(c))$, which implies $p^{C}_{C \cap D} \neq p^{D}_{C \cap D}$. Moreover $p$ is non-degenerate because $p^{C}_{b} = p_{b} = p^{D}_{b}$ and $p^{C}_{c} = p_{c} = p^{D}_{c}$. For example, for $b$ we have
\begin{align}
    p^{C}_{b}(o) &= \sum_{\begin{subarray}{l} s \in O^{C} \\ s(b) = 0\end{subarray}}p^{C}(s) =  \sum_{\begin{subarray}{l} r \in O^{C \cap D} \\ r(b) = o \end{subarray}}\sum_{\begin{subarray}{l} s \in O^{C} \\ s|_{C \cap D}  = r \end{subarray}}p^{C}(s)
    \\
    &= \sum_{\begin{subarray}{l} r \in O^{C \cap D} \\ r(b) = o \end{subarray}}\sum_{\begin{subarray}{l} s \in O^{C} \\ s|_{C \cap D}  = r \end{subarray}}p_{a}(s(a))f(s(b),s(c))
    \\
    &=\sum_{\begin{subarray}{l} r \in O^{C \cap D} \\ r(b) = o \end{subarray}}f(r(b),r(c)) = p_{b}(o)
    \\
    &=  \sum_{\begin{subarray}{l} r \in O^{C \cap D} \\ r(b) = o \end{subarray}} g(r(b),r(c))
    \\
    &= \sum_{\begin{subarray}{l} r \in O^{C \cap D} \\ r(b) = o \end{subarray}}\sum_{\begin{subarray}{l} t \in O^{D} \\ t|_{C \cap D} = r\end{subarray}}  g(t(b),t(c))p_{d}(t(d))
    \\
    &= p^{D}_{b}(o).
\end{align}
Hence the following lemma has been proved.

\begin{lemma}\label{NDeg is not ND}
Non-degenerate behaviours are \textbf{not} necessarily non-disturbing.
\end{lemma}

Non-disturbance seems to be a condition way stronger then non-degeneracy: why should a non-degenerate behaviour, whose component $p^{C}$ is any coupling for $p_{x}$, $x \in C$, give always the same coupling in intersections $p_{x}$, $x \in C \cap D$? 

\section{Non-contextuality}\label{SecNoncontextuality}

Before we begin our discussion we want to emphasize an important (and quite natural) property of couplings. Roughly speaking, it is the fact that marginals distributions of couplings are also couplings. Let $p_{1}, ..., p_{n}$ be probability distributions in a finite set $O$ and $p: O^{n} \ri [0,1]$ be a coupling of it. Now take any subsequence $p_{i_{k}}$, $k=1, ..., m$, of $p_{i}$, $i=1, ..., n$. The marginal of $p$ in $O^{\{i_{1}, ..., i_{m}\}} \cong O^{m}$, that is, the function $q: O^{m} \ri [0,1]$ given by 
$$q(s) \doteq \sum_{\begin{subarray}{l}u \in O^{n} \\ \forall k=1,...,m(u_{i_{k}} = s_{k})\end{subarray}} p(u),$$
is a coupling of $p_{i_{1}}, ..., p_{i_{m}}$. Consequently, if $S \equiv (S_{1}, ..., S_{n})$ is a coupling of  a sequence of random variables $R_{1}, ..., R_{n}$,  $S' \equiv (S_{i_{1}}, ..., S_{i_{m}})$ is a coupling of $R_{i_{1}}, ..., R_{i_{m}}$. 

Now let $(\mathcal{X},\mathcal{C},\mathscr{R})$ by a system in a scenario $\mathcal{S} \equiv (\mathcal{X},\mathcal{C},O)$ and, for any $x \in \mathcal{X}$ and $C \in \mathcal{C}_{x}$, let $p^{C}_{x}:O \ri [0,1]$ be the distribution associated to the random variable $R^{C}_{x}$ of the system. The system $(\mathcal{X},\mathcal{C},\mathscr{R})$ has a maximally non-contextual description iff there is a coupling $q$ of $\underline{P} \doteq (p
^{C}_{x}| \forall x \in \mathcal{X} (C \in \mathcal{C}_{x}))$ such that, for any $x \in \mathcal{X}$, its marginal (in the appropriate coordinates) is a maximal coupling of $\underline{P}_{x} \doteq (p^{C}_{x}| C \in \mathcal{C}_{x})$. Note that $q$ defines a behaviour $\widetilde{q}$ for this system, where $\widetilde{q}_{C}$ is given by the marginal of $q$ in the components associated to the distributions $p^{C}_{x}$, $x \in C$. We introduce the following definition.

\begin{dfn}[Extended contextuality]\label{DefinitionNoncontextual} We call non-contextual in the extended sense a behaviour whose system, as defined by it (see definition \ref{DefinitionForNotation}), has a maximally non-contextual description, that is, a behaviour $p$ in a scenario $\mathcal{S} \equiv(\mathcal{X},\mathcal{C},O)$ is non-contextual in the extended sense when $\underline{P} \doteq (p
^{C}_{x}|\forall x \in \mathcal{X} ( C \in \mathcal{C}_{x}))$ has a coupling $q$ satisfying the following properties:
\begin{itemize}
    \item[(a)] For any $C \in \mathcal{C}$, 
    $$q_{\widetilde{C}} = p^{C},$$ 
    where $q_{\widetilde{C}}$ is given by the marginal of $q$ in the coordinates corresponding to $p^{C}_{x}$, $x \in C$
    \item[(b)] The distribution $q_{T(x)}$ is a maximal coupling of $\underline{P}_{x} \doteq (p^{C}_{x}| C \in \mathcal{C}_{x})$, where $q_{T(x)}$ is given by the marginal of $q$ in the coordinates corresponding to $p^{C}_{x}$, $C \in \mathcal{C}_{x}$.
    \end{itemize}
\end{dfn}

Now we want to show that the so called standard definition of contextuality \cite{CSW2014,amaral2018graph} (definition \ref{DefNoncontextualStandard}) is the above definition when restricted to non-degenerate behaviours (lemma \ref{LemmaStandardContextuality}), as it has already been discussed in \cite{kujala2015necessary}. In order to do that it is useful to fix some notation. First of all, if a behaviour $p$ in $\mathcal{S}$ is non-contextual in the above sense then $\underline{P}$ has a coupling $q:O^{n} \ri [0,1]$, where $n = \vert \cup_{C \in \mathcal{C}}\{(x,C); x \in \mathcal{C}\}\vert$. Let's denote by $x_{C}$ the component $i \in \{1, ..., n\}$ associated to $p^{C}_{x}$, that is, $i$ is such that $p^{C}_{x}(o) = \sum_{\begin{subarray}{l} u \in O^{n} \\ u_{i} = o\end{subarray}}q(u) \ \forall o \in O$. We have $O^{n} \doteq O^{\{1, ..., n\}} = O^{\widetilde{X}}$, where 
$$\widetilde{\mathcal{X}} \doteq \bigcup_{C \in \mathcal{C}} \{x_{C}; x \in C\}.$$ 
We also define $\widetilde{C} \doteq \{x_{C} \in \widetilde{\mathcal{X}}; x \in C\}$ for every $C \in \mathcal{C}$ and $T(x) \doteq \{x_{C} \in \widetilde{\mathcal{X}}; C \in \mathcal{C}_{x}\}$ for every $x \in \widetilde{\mathcal{X}}$. This justifies the notation of definition \ref{DefinitionNoncontextual}. The following embedding will be very important in our discussion.

\begin{align}
    &O^{\mathcal{X}} \ni u \hookrightarrow \widetilde{u} \in O^{\widetilde{\mathcal{X}}}
    \\
    &\forall x_{C} \in \widetilde{\mathcal{X}} (\widetilde{u}(x_{C}) \doteq u(x)).
\end{align}

We will denote by $C(O^{\widetilde{X}})$ the ``copy'' of $O^{\mathcal{X}}$ in $O^{\widetilde{\mathcal{X}}}$, that is,  $C(O^{\widetilde{\mathcal{X}}})\equiv \{\widetilde{u}; u \in O^{\mathcal{X}}\}$. Here, the symbol $C$ means ``constant'', because any $\widetilde{u} \in O^{\widetilde{\mathcal{X}}}$ is constant in the set $T(x)$ for every $x \in \mathcal{X}$. Now we note the following: Any probability distribution $q:O^{\mathcal{\mathcal{X}}} \ri [0,1]$ has a trivial extension  $\widetilde{q}: O^{\widetilde{\mathcal{X}}} \ri [0,1]$ given by
\begin{align}
\widetilde{q}(\widetilde{u}) &\doteq q(u) \ \text{for any} \ u \in O^{\mathcal{X}}
\\
q(v) &\doteq 0 \ \text{if} \ v \notin C(O^{\widetilde{\mathcal{X}}}).
\end{align}
Moreover, any distribution $p:O^{\widetilde{\mathcal{X}}} \ri [0,1]$ satisfying $p(v) = 0$ for all $v \notin C(O^{\widetilde{\mathcal{X}}})$ is the extension (in the above sense) of one, and only one, distribution $q:O^{\mathcal{X}} \ri [0,1]$, being $q$ given by $q(u) \doteq p(\widetilde{u}) \ \forall u \in O^{\widetilde{X}}$. Note that, if $q:O^{\mathcal{X}} \ri [0,1]$, for any $x \in \mathcal{X}$ and any $C \in \mathcal{C}_{x}$ we have $q_{x} = \widetilde{q}_{x_{C}}$ (where $\widetilde{q}:O^{\widetilde{\mathcal{X}}} \ri [0,1]$ is the trivial extension of $q$). Finally, 
\begin{align}\label{Eq1}
    \widetilde{q}_{\widetilde{C}}(s) &= \sum_{\begin{subarray}{l} v \in O^{\widetilde{\mathcal{X}}} \\ v|_{\widetilde{C}} = s\end{subarray}} \widetilde{q}(v) = \sum_{\begin{subarray}{l} v \in C(O^{\widetilde{\mathcal{X}}}) \\ v|_{\widetilde{C}} = s\end{subarray}} \widetilde{q}(v) = \sum_{\begin{subarray}{l} u \in O^{\mathcal{X}} \\ u|_{C} = s\end{subarray}} q(u),
\end{align}
for any $C \in \mathcal{C}$. For $x \in \mathcal{X}$ we have
\begin{align}\label{Eq2}
    \widetilde{q}_{T(x)}(s) &= \sum_{\begin{subarray}{l} v \in O^{\widetilde{\mathcal{X}}} \\ v|_{T(x)} = s\end{subarray}} \widetilde{q}(v) = \sum_{\begin{subarray}{l} v \in C(O^{\widetilde{\mathcal{X}}}) \\ v|_{T(x)} = s\end{subarray}} \widetilde{q}(v).
\end{align}
This last result implies that $\widetilde{q}_{T(x)}(s) \neq 0$ only if $s$ is constant, and in this case
\begin{align}\label{Eq3}
    \widetilde{q}_{T(x)}(c_{o}) &=  \sum_{\begin{subarray}{l} v \in C(O^{\widetilde{\mathcal{X}}}) \\ v|_{T(x)} = c_{o}\end{subarray}} \widetilde{q}(v) = \sum_{\begin{subarray}{l} u \in O^{\mathcal{X}} \\ u(x) = o \end{subarray}} q(v).
\end{align}
This means that $\widetilde{q}_{T(x)}$ is a maximal coupling (the only one) of $q_{x}, ..., q_{x}$ ($\vert \mathcal{C}_{x}\vert$ times) - remember that $q_{x} = \widetilde{q}_{x_{C}}$.

Now we prove the following lemma.

\begin{lemma}\label{LemmaStandardContextuality}
Let $p$ be a behaviour in a scenario $\mathcal{S} \equiv (\mathcal{X},\mathcal{C}, O)$. The following statements are equivalent.
\begin{itemize}
    \item[(a)] The behaviour $p$ is non-contextual in the standard sense.
     \item[(b)] The behaviour $p$ is non-contextual in the extended sense (definition \ref{DefinitionNoncontextual}), and non-degenerate.
    
\end{itemize}
\end{lemma}

\begin{proof}
We know that a ``non-contextual in the standard sense'' behaviour is non-disturbing, therefore it is non-degenerate (lemma \ref{LemmaNondisturbingNondegenerate}); then what we have to prove is that a non-degenerate behaviour is non-contextual in the extended sense (definition \ref{DefinitionNoncontextual}) iff it is non-contextual in the standard sense. So let $p$ be a non-degenerate behaviour in $\mathcal{S}$. We have $P = \{p_{x}, x \in \mathcal{X}\}$, in particular $P_{x} = \{p_{x}\}$, that is, in $\underline{P}_{x} = (p^{C}_{x}; C \in \mathcal{C}_{x})$ we are just repeating  $p_{x}$ $\vert \mathcal{C}_{x}\vert$ times. Suppose that $p$ is non-contextual in the standard sense and let $q:O^{\mathcal{X}} \ri [0,1]$ be a global section of $p$, that is, a function satisfying, for any $C \in \mathcal{C}$, $p^{C}(s)  = \sum_{\begin{subarray}{l} u \in O^{\mathcal{X}} \\ u|_{C} = s \end{subarray}} q(u)$. We have already seen that the trivial extension $\widetilde{q}:O^{\widetilde{\mathcal{X}}} \ri [0,1]$ of $q$  is a coupling for $\underline{P}$. Moreover, equation \ref{Eq1} imply that $\widetilde{q}_{\widetilde{C}} = p^{C}$ and equations \ref{Eq2} and \ref{Eq3} implies that $\widetilde{q}_{T(x)}$ is a maximal coupling of $p_{x}, ..., p_{x}$ ($\vert \mathcal{C}_{x}\vert$ times), that is, conditions $a$ and $b$ of definition \ref{DefinitionNoncontextual} are satisfied, implying that $p$ is non-contextual (definition \ref{DefinitionNoncontextual}). On the the other hand, if $f:O^{\widetilde{X}} \ri [0,1]$ is a coupling for $\underline{P} = (p^{C}_{x};\forall x \in \mathcal{X} (C \in \mathcal{C}_{x}))$ follows that $f_{T(x)}(s) = 0$ whenever $s \in O^{T(x)}$ is not constant (see appendix \ref{appendix: couplings}), but in this case
\begin{align}
    f_{T(x)}(s) = \sum_{\begin{subarray}{l} v \in O^{\widetilde{\mathcal{X}}} \\ v|_{T(x)} = s \end{subarray}}f(v).
\end{align}
This implies that $f(v) = 0$ whenever the restriction $v|_{T(x)}$ is not constant. This is true for any $x \in \mathcal{X}$, therefore $f(v) = 0$ if $v \notin C(O^{\widetilde{\mathcal{X}}})$. But we have seen that this implies $f = \widetilde{q}$ for one (and only one) $q:O^{\mathcal{X}} \ri [0,1]$, and this function is a global section for $p$, proving that $p$ is non-contextual in the standard sense.
\end{proof}

\section{extended scenarios and behaviours}\label{SecExtended}
Now we want to formulate definition \ref{DefinitionNoncontextual} using the definition of so-called extended scenarios, that were introduced in \cite{dzhafarov2013allpossiblecouplings,kujala2015necessary}.

Let $p$ be a non-contextual behaviour in $\mathcal{S} \equiv (\mathcal{X},\mathcal{C},O)$ and let $q:O^{n} \ri [0,1]$ be a coupling of $\underline{P}$. In the above section we have denoted by $x_{C}$ the component $i \in \{1, ..., n\}$ associated to $p^{C}_{x}$, which allows us to write $\widetilde{X} \equiv \{1, ..., n\}$ and define $\widetilde{C} \doteq \{x_{C} \in \widetilde{X}; x \in C\}$ for $C \in \mathcal{C}$ and $T(x) \doteq \{x_{C}  \in \widetilde{\mathcal{X}}; C \in \mathcal{C}_{x}\}$ for any $x \in \mathcal{X}$. The triple $\widetilde{\mathcal{S}} \equiv (\widetilde{\mathcal{X}}, \widetilde{\mathcal{C}}, O)$, where $\widetilde{\mathcal{C}} = \{\widetilde{C}; C \in \mathcal{C}\} \cup \{T(x); x \in \mathcal{X} \wedge \vert \mathcal{C}_{x}\vert > 1\}$, is a scenario (note that $\widetilde{X} = \cup_{C \in \mathcal{C}} \widetilde{C}$). In this scenario, a coupling  $q: O^{\widetilde{\mathcal{X}}} \ri [0,1]$ of $\underline{P}$ can be seen as a global section of a behavior $\widetilde{p}$ in $\widetilde{\mathcal{S}}$, being $\widetilde{p}$ given by $\widetilde{p}^{\widetilde{C}} \doteq q_{\widetilde{C}} = p^{C}$ for $C \in \mathcal{C}$ and $\widetilde{p}^{T(x)} \doteq q_{T(x)}$ for $x \in \mathcal{X}$. Note that the behaviour $\widetilde{p}$ depends on $q$, not just on $p$. These behaviours $\widetilde{p}$  are examples of the so-called ``extensions of $p$ in $\widetilde{\mathcal{S}}$'' \cite{amaral2018necessaryextended}, an we conclude that a behaviour $p$ in $\mathcal{S}$ is non-contextual iff there is an extension $\widetilde{p}$ of $p$ which is non-contextual in the standard sense (that is, has a global section) in $\widetilde{\mathcal{S}}$. Let's introduce it in a more explicit way.

\begin{dfn}[Extension of a scenario] Let $\mathcal{S} \equiv (\mathcal{X},\mathcal{C},O)$ be a scenario. For any pair $x \in \mathcal{X}$ and $C \in \mathcal{C}$ we define $x_{C} \doteq (x,C)$; note that $x_{C} = y_{D}$ iff $x = y$ and $C = D$. We define, for any $C \in \mathcal{C}$,
    $$\widetilde{C} \doteq \{x_{C}; x \in C\};$$
   and for any $x  \in \mathcal{C}$ satisfying $\vert \mathcal{C}_{x} \vert >1$,
    $$T(x) \doteq \left\{x_{C}; C \in \mathcal{C}_{x}\right\}.$$
    The ``extension of $\mathcal{S}$'' is the scenario $\widetilde{\mathcal{S}} \equiv (\widetilde{\mathcal{X}},\widetilde{\mathcal{C}}, O)$, where
    $$\widetilde{X} \doteq \bigcup_{C \in \mathcal{C}} \widetilde{C}$$
    and
    $$\widetilde{\mathcal{C}} \doteq \left\{\widetilde{C}; C \in \mathcal{C}\right\} \cup \left\{T(x); x \in \mathcal{X} \wedge \vert \mathcal{C}_{x} \vert > 1\right\}.$$
\end{dfn}
    There is a trivial isomorphism $O^{C} \ni s \xmapsto{\sim} \widetilde{s} \in O^{\widetilde{C}}$, where $\widetilde{s}(x_{C}) = s(x) \ \forall x \in C$. From now on we will write $s \equiv \widetilde{s}$. Moreover, if $f:O^{C} \ri [0,1]$ and $g: O^{\widetilde{C}} \ri [0,1]$ satisfies $g(\widetilde{s}) = f(s) \forall s \in O^{C}$, that is, if $f = g \circ \sim$, for the sake of simplicity we will write $f = g$.
    
    \begin{dfn}[Extension of a behaviour] Let $p$ be a behaviour in $\mathcal{S} \equiv (\mathcal{X},\mathcal{C},O)$. We call an ``extension of $p$'' any behaviour $f$ in $\widetilde{S}$ satisfying:
    \begin{itemize}
        \item[(a)] For any $C \in \mathcal{C}$,
        $$f^{\widetilde{C}} = p^{C}.$$
        \item[(b)] For any $x \in \mathcal{X}$ satisfying $\vert \mathcal{C}_{x} \vert > 1$, $f^{T(x)}$ is a maximal coupling of $\underline{P}_{x} \doteq \left (p^{C}_{x}| C \in \mathcal{C}_{x}\right)$.
    \end{itemize}
    \end{dfn}
    
    The following lemma has already been proved above. This says that definition \ref{DefinitionNoncontextual} coincides with the definition of extended contextuality introduced in \cite{amaral2018necessaryextended}.
    
    \begin{lemma}
    Let $p$ be a behaviour in $\mathcal{S}$. The following are equivalents:
    \begin{itemize}
        \item[(a)] $p$ is non-contextual with respect to definition \ref{DefinitionNoncontextual}
        \item[(b)] $p$ has an extension $f$ which is non-contextual in the standard sense, i.e., $f$ has a global section in $\widetilde{S}$.
    \end{itemize}
    \end{lemma}

    Finally, let's characterize extensions of non-degenerate behaviours. It has been proved in the appendix that a sequence of probability distributions $p_{1}, ..., p_{n}:O \ri [0,1]$ has a coupling $q:O^{n} \ri [0,1]$ satisfying $\sum_{o \in O} p(c_{o}) = 1$, where $c_{o} = (o,o,...,o)$, iff $p_{1}= ...= p_{n}$; note that $q(c_{o}) = p_{i}(o) \forall o \in O$ and, moreover, in this case $q$ is the only maximal coupling of $p_{1}, ..., p_{n}$. This result immediately implies the following lemma.

    \begin{lemma}
    Let $p$ be a behaviour in $\mathcal{S}$. The following statements are equivalent.
    \begin{itemize}
        \item[(a)] $p$ is non-degenerate.
        \item[(b)] $p$ has only one extension $\widetilde{p}$, which satisfies 
        $$\widetilde{p}_{T(x)}(x_{C_{1}} = ... = x_{C_{\vert \mathcal{C}_{x}\vert}}) = 1$$
        for every $x \in \mathcal{X}$, where
        $$\widetilde{p}_{T(x)}(x_{C_{1}} = ... = x_{C_{\vert \mathcal{C}_{x}\vert}})\doteq \sum_{o \in O} p(c_{o}).$$
    \end{itemize}
    \end{lemma}

    Note that ``$\widetilde{p}_{T(x)} = p_{x}$'', that is, $\widetilde{p}_{T(x)} \circ \iota = p_{x}$, where $o \in O \xmapsto{\iota} c_{o} \in O^{T(x)}$. Roughly speaking, we can say that the extension of a non-degenerate behaviour $p$ is just the set of distributions $P_{x} = \{p_{x}; x \in \mathcal{X}\}$ ``plus'' the set of contextual couplings $p^{C}$, $C \in \mathcal{C}$, given by $p$.

\section{Conclusion}

In our work we have studied the formalism of contextuality-by-default from the perspective of the contextuality scenario paradigm. From that we were able to rephrase the notion of consistent connectedness in terms of a new definition associated to the behaviour perspective: non-degenerate behaviours. We have shown that non-degeneracy and non-disturbance are different notions. More important, such results imply that the contextuality approach is compatible with the idea of measurements as a collection of random variables, relaxing the condition of non-disturbance over the behaviours present in the experimental scenario.

From such perspective we could introduce in the already known scenario formalism the same ideas relevant for the contextuality-by-default approach. Hence obtaining an extended notion of non-contextuality that does not require non-disturbance from the behaviours, suitable for the interpretation of measurement procedures in different contexts as truly different procedures; but arising from (graph) contextuality scenarios, framework presented in \cite{CSW2014,amaral2018graph}. As an observation we further noticed that the set of non-degenerate behaviours form a polytope. 

Another important remark is that we have proved that non-contextuality in the usual sense is equivalent to the extended notion of non-contextuality when restricted to non-degenerate behaviours. 

Our work brings light to the interplay between contextuality approaches, but there are still further directions to investigate. An interesting research can be the development of a resource theory for extended contextuality that could investigate contextual advantages beyond the non-disturbance condition. Another possible investigation would be to understand the connection between the generalized contextuality approach and the contextuality-by-default approach, since these are two very different perspectives on the notion of a context: equivalent probabilities for the former and different measurement procedures for the latter.

\section*{Acknowledgments}
The authors would like to thank the programs CAPES and CNPq for the financial support. We also would like to thank the Quantum Foundations Group in the Physics Institute of the University of São Paulo, with special thanks to Giulio H. C. da Silva, for helpful discussions.

\bibliographystyle{apsrev4-1}
\bibliography{Bibliography}

\appendix
\section{Appendix}\label{appendix: couplings}

\subsection{Couplings}
Define $(\Omega,\Sigma,\muu)$ as a probability space and $(\tilde{\Omega},\tilde{\Sigma})$ as a measurable space, a \textbf{ random variable in} $(\Omega,\Sigma,\muu)$ \textbf{with values on} $(\tilde{\Omega},\tilde{\Sigma})$ is a measurable function $R: \Omega \ri \tilde{\Omega}$ (more precisely, a $(\tilde{\Omega},\tilde{\Sigma})-(\Omega,\Sigma)$ measurable function). A \textbf{distribution measure of a random  variable} is defined as follows: for a random variable $R: \Omega \ri \tilde{\Omega}$, its distribution measure is the measure $\muu_{R}$ on $(\tilde{\Omega}, \tilde{\Sigma})$ given by $\muu_{R}(A) \doteq \muu(R^{-1}(A))$ for any $A \in \tilde{\Sigma}$; in other words, it is the pushforward measure defined by $R$ on $(\tilde{\Omega},\tilde{\Sigma})$. 

In this paper we are interested  in discrete random variables and in measures defined in discrete measurable spaces \cite{roch2015modern}; this simplifies our discussion: If $O$ is a finite set, there is a well known one-to-one correspondence between probability measures on $(O,\mathcal{P}(O))$, were $\mathcal{P}(O)$ is the power set, and probability distributions in $O$ (functions $p: O \ri [0,1]$ satisfying $\sum_{o \in O} p(o) = 1$). Hence, instead of considering the distribution measure of a random variable $R$ in $(\Omega, \Sigma,\muu)$ with values on $(O,\mathcal{P}(O))$ we will consider its density function, which is the unique distribution $p$ in $O$ satisfying $\muu(R^{-1}(A)) = \sum_{o \in A}p(o)$ for any $A \in \mathcal{P}(O)$. We will denote by $p_{R}$ the density function of $R$ and call it the \textbf{distribution} of $R$. 

Any probability distribution $p$ in a finite set $O$ is the distribution of a random variable with values on $(O,\mathcal{P}(O))$; a simple example of random variable whose distribution is $p$ is the identity function $I: O \ri O$, which is a random variable in $(O,\mathcal{P}(O),\muu^{p})$ with values on $(O,\mathcal{P}(O))$, where $\muu^{p}$ is the measure associated to $p$, that is, $\muu^{p}(A) = \sum_{o \in A} p(o)$ for any $A \in \mathcal{P}(O)$. We know that there is not a one-to-one correspondence between random variables with values on $(O,\mathcal{P}(O))$ and distributions in $O$. However, any distribution induces an equivalence class of random variables in the following sense: if $R,S$ are random variables with values on $(O,\mathcal{P}(O))$, where $O$ is finite, we call this variables \textbf{physically equivalent} whenever they have the same distribution. The adjective ``physically'' comes from the fact that, from an operationalist point of view, what we obtain in a laboratory is just the probability distributions over $O$, the set of outcomes of our measurements; the random variable we associate to the distribution we have obtained is any random variable whose distribution matches with it. Notice that a random variable $R$ with values on $(O,\mathcal{P}(O))$ is always physically equivalent to the identity when we consider it as a random variable in $(O,\mathcal{P}(O),\muu_{R})$ with values on $(O,\mathcal{P}(O))$. This means that the probability space $(\Omega,\Sigma,\muu)$ and even the random variable $R$ are superfluous in some sense: in the laboratory we obtain a probability distribution $p$ in $O$ and we can always associate to it the random variable $I:O \ri O$ in the probability space $(O,\mathcal{P}(O),\muu^{p})$ (that is, the events are exactly what we intuitively would call events).

We will define, in what follows, \textbf{coupling of a finite sequence of probability distributions}: Let $p_{1}, ... p_{n}$ be probability distributions over a finite set $O$. A coupling for $p_{1},..., p_{n}$ is a probability distribution $q:O^{n} \ri [0,1]$ satisfying, for any $i \in \{1, ..., n\}$,
$$p_{i}(o) = \sum_{\begin{subarray}{l} s \in O^{n} \\ s_{i} = o\end{subarray}} q(s),$$
that is, $q:O^{n} \ri [0,1]$  has $p_{1}, ..., p_{n}$ as marginal distributions. Couplings of a finite sequence of distributions with the same domain always exists, e.g. the product of the distributions, that is, the function $q(s) \doteq \prod_{i=1}^{n} p_{i}(s_{i})$, is a possible coupling for them.

We can also define \textbf{coupling of a finite sequence of discrete random variables}. For $i=1,..., n$, let $R_{i}$ be a random variable in $(\Omega_{i}, \Sigma_{i}, \muu_{i})$ with values on $(O,\mathcal{P}(O))$ - it's important that all the random variables have the same co-domain. A coupling of $R_{1}, ..., R_{n}$ is a random variable $R$ with values on $(O^{n},\mathcal{P}(O^{n}))$ such that $p_{R}:O^{n} \ri [0,1]$ is a coupling of $p_{R_{1}}, ..., p_{R_{n}}$. Couplings of a finite sequence of random variables always exists. An example of coupling of $R_{1}, ..., R_{n}$ is the product $R \equiv (R_{1}, ..., R_{n})$, which is a random variable in $(\Omega_{1}, \times \dots \times \Omega_{n}, \Sigma_{1} \otimes \dots \otimes \Sigma_{n},\muu_{1} \otimes \dots \otimes \muu_{n})$ with values on $(O^{n},\mathcal{P}(O^{n}))$; in this case $p_{R}(s)= \prod_{i=1}
^{n}p_{R_{i}}(s_{i})$. We can also understand a coupling of $R_{1}, ..., R_{n}$ as a sequence of random variables $S_{1}, ..., S_{n}$ in $(\Omega,\Sigma,\muu)$ (any probability space) with values on $(O,\mathcal{P}(O))$ such that, for any $i=1, ..., n$, $R_{i}$ and $S_{i}$ have the same distribution (are physically equivalent); from this point  of view, what is important about couplings it that all these random variables are defined in the same probability space. It is easy to justify this ``equivalent notions'' of couplings. On the one hand, the product $S$ of such a collection is a coupling for $R_{1}, ..., R_{n}$ (notice that the product is a random variable in the same probability space as $S_{i}$ just because all this random variables are defined in the same probability space). On the other hand, denoting by $(\Omega,\Sigma,\muu)$ the probability space in which a coupling $R$ of $R_{1}, ..., R_{n}$ is defined, we see that $R$ necessarily is the product of random variables $S_{1}, ..., S_{n}$ in $(\Omega,\Sigma,\muu)$ with values on $(O,\mathcal{P}(O))$ because $R$ is given by components; moreover, $p_{S_{i}} = p_{R_{i}}$. Finally, it is useful to notice that, for a coupling $S$ of $R_{1}, ..., R_{n}$, the distribution $p_{S}$ is the product $\Pi_{i=1}^n p_{R_{i}} $ iff its components $S_{1}, ..., S_{n}$ are independent random variables. 

Let $R_{1}, ..., R_{n}$ be random variables with values on $(O,\mathcal{P}(O))$, where $O$ is a finite set. For $o \in O$ we denote by $c_{o}$ the $n$-uple $(o,o, ..., o)$, and we define the set $E \doteq \{c_{o} \in O^{n}; o \in O\}$. Now Let $S \equiv (S_{1}, ..., S_{N})$ be any coupling of $R_{1}, ..., R_{n}$. In the probability space $(O,\mathcal{P}(O),\muu_{S})$, $c_{n}$ can be understood as the simple event `` all the random variables $S_{1},...,S_{n}$ have assumed the value $o$'', while $E$ can be understood as ``all the random variables $S_{1},...,S_{n}$ have assumed the same value''. A \textbf{maximal coupling} of $R_{1}, ..., R_{n}$ is a coupling which maximize the probability of $E$, that is, if $M$ is a maximal coupling of   $\{R_{1}, ..., R_{n}\}$ then $\muu_{M}(E) \geqslant \muu_{S}(E)$ for any coupling $S$ of $\{R_{1}, ..., R_{n}\}$. It can be proved that such coupling always exists, although it is not necessarily unique \cite{amaral2019extendedcontextuality,amaral2018necessaryextended}.  Analogously - and consistently - we define maximal coupling of probability distributions $p_{1}$, ..., $p_{n}$: it is a coupling $p$ of $p_{1}, ..., p_{n}$ satisfying $p(E) \geqslant q(E)$ for any other coupling $q$ of $p_{1}, ..., p_{n}$.

 We conclude the appendix with an intuitive result: a sequence $p_{1}, ..., p_{n}$ satisfies $p_{1}=...=p_{n}$ iff there is a coupling $q$ of $p_{1},...,p_{n}$ such that $q(E)=1$; moreover, in this case $q$ is the only maximal coupling of $p_{1}=...=p_{n}$. Let's prove it. If $p \equiv p_{1}=...=p_{n}$ we define $q:O^{n} \ri [0,1]$ by $q(c_{o}) \doteq p(o) \ \forall o \in O$, which implies $p(s)=0$ when $s \notin E$, and this function is a  coupling of $p_{1}, ..., p_{n}$ satisfying $q(E) = 1$. On the other hand, if a coupling $q$ of a sequence $p_{1}, ..., p_{n}$ satisfies $q(E) = 1$ then $q(s) = 0$ whenever $s \notin E$, therefore $p_{i}(o) = \sum_{\begin{subarray}{l} s \in O^{n} \\ s_{i} = 0\end{subarray}} = q(c_{o})$, and this implies $p_{1}=...=p_{n}$. Finally, $q$ is the only maximal coupling of $p_{1},...,p_{n}$ because any other maximal coupling $q'$ satisfies $q(c_{o}) = p_{i}(o) = \sum_{\begin{subarray}{l} s \in O^{n} \\ s_{i}=0 \end{subarray}}q'(s) = q'(c_{o})$, which implies $q'=q$.

\end{document}